\documentclass[12pt]{article}
\usepackage[dvips]{graphicx}
\usepackage{amsmath,amssymb,dsfont}

\newtheorem{theorem}{Theorem}[section]
\newtheorem{proposition}[theorem]{Proposition}

\newtheorem{corollary}[theorem]{Corollary}
\newtheorem{remark}[theorem]{Remark}
\newtheorem{example}{Example}
\newenvironment{proof}
{\vspace*{3mm} \noindent {\sc Proof~: }\rm  }{$\emptysq$ \vspace*{3mm}}
\renewcommand{\theequation}{\thesection.\arabic{equation}} 

\def\emptysq{\mathbin{\vbox{\hrule\hbox{\vrule height1ex \kern.5em 
                         \vrule height1ex}\hrule}}}
\def\finl{\hfill\break}
\def\build#1_#2^#3{\mathrel{
\mathop{\kern 0pt#1}\limits_{#2}^{#3}}}
\def\tend_#1^#2{\mathrel{
\mathop{\kern 0pt\longrightarrow}\limits_{#1}^{#2}}}
\def\ntinf{n \rightarrow \infty}
\def\tendvers{\tend_{\ntinf}^{}}
\def\tendas{\tend_{\ntinf}^{a.s.}}

\def\tenddist{\tend_{\ntinf}^{{\cal L}}}

\def\egalas{\build{=}_{}^{{\rm a.s.}}}

\def\norm#1{\|#1 \|}
\def\normp#1{\parallel #1 \parallel}
\def\abs#1{\left| #1 \right|}

\def\acc#1{\left\{ #1 \right\}}
\def\pa#1{\left( #1 \right)}
\def\pab#1{\bigl( #1 \bigr)}
\def\paB#1{\Bigl( #1 \Bigr)}
\def\crob#1{\Bigl[ #1 \Bigr]}
\def\cro#1{\left[ #1 \right]}

\def\bkE{{\mathbb{E}}}
\def\bkF{\mbox{{\rm I\kern-.17em F}}}
\def\bkone{\mbox{{\rm 1\kern-.25em I}}}
\def\bkN{{\mathbb{N}}}
\def\bkR{{\mathbb{R}}}
\def\bkRd{{\mathbb{R}^d}}

\def\veps{\varepsilon}

\def\fchap{\widehat f_n}

\def\pchap{\widehat p_n}

\def\dt{\mbox{\rm d}t}

\def\dv{\mbox{\rm d}v}

\def\ind{\mathds{1}}

\font\calcal=cmsy10 scaled\magstep1
\def\build#1_#2^#3{\mathrel{\mathop{\kern 0pt#1}\limits_{#2}^{#3}}}
\def\liml{\build{\longrightarrow}_{\ntinf}^{{\mbox{\calcal L}}}}
\newcommand{\CMx}{<\!M(x)\!>}

\def\wunn{w_{1,n}}
\def\wdeuxn{w_{2,n}}
\def\indEi{\ind_{\acc{\left\|X_i\right\| \le v_i}}}
\def\indEic{\ind_{\acc{\left\|X_i\right\| > v_i}}}

\title{Strong uniform consistency and asymptotic normality of a kernel based error density estimator in functional autoregressive models}

\author{Nadine HILGERT\thanks{UMR 729 MISTEA, INRA SupAgro, 2 Place Viala, 34060 Montpellier Cedex, France} and Bruno PORTIER\thanks{LMI EA 3226, 
INSA de Rouen, Place Emile Blondel - BP 8, 76131 Mont-Saint-Aignan Cedex, France}}

\date{}

\begin{document}

\maketitle

\begin{abstract}
Estimating the innovation probability density is an important issue in any regression analysis.
This paper focuses on functional autoregressive models. A residual-based kernel estimator is proposed for the innovation density. Asymptotic
properties of this estimator depend on the average prediction error of the functional autoregressive function. Sufficient conditions are studied
to provide strong uniform consistency and asymptotic normality of the kernel density estimator.
\end{abstract}

\vspace{0.2cm}


\vspace{0.2cm}


\pagestyle{myheadings}
\thispagestyle{plain}
\markboth{B. PORTIER}{Density Estimation in Functional Autoregressive models}

\section{Introduction}
\setcounter{equation}{0} 
Dealing with regression estimation procedure gives rise to important questions concerning the a posteriori diagnostic of model assumptions. Diagnostic tools are generally based on the residuals. For example, one may have to check if the innovations are Gaussian ones. This is required in the context of variable selection or model change detection among others.
Checking such an assumption may involve estimating the innovation density 
and investigating the 
asymptotic convergence properties of the estimate. Kernel-based methods are among 
the most common nonparametric methods used to that purpose.
Since the pioneer works of Rosenblatt\,\cite{Rosenblatt} and Parzen\,\cite{Parzen},
a wide range of literature is available on kernel density estimation.
We refer the reader to \cite{Devroye1}, \cite{Devroye2}, \cite{Silverman}
for some interesting books on density estimation
in the context of the independent and identically distributed sample,
mixing processes, etc.
However, few papers investigate the asymptotic properties of 
a kernel density estimator (KDE for short) associated with the driven noise in  
a given regression or autoregressive model.

\bigskip

When dealing with such models, the driven noise is not observed. 
Its probability density function (pdf for short) can only be estimated through the residual error calculated from the estimation of the unknown component of the model. This one 
shall thus be estimated with an appropriate convergence rate to induce good properties to the residual error. 
A common noise density estimator is the Parzen-Rosenblatt kernel estimator, based on this residual error which is then considered as a noise predictor. 

Chai et al. \cite{Chai-Li-Tian} proved the uniform strong consistency on $\bkR$ of the noise KDE in the linear regression case. 

The linear parametric autoregressive case is for example studied in Koul \cite{Koul92} who gave weak convergence results, or Cheng \cite{Cheng05} who also showed that the asymptotic distribution of the maximum of a suitably normalized deviation of the
density estimator from the expectation of the kernel error density (based on
the true error) is the same as in the case of the one sample set up, which is 
given in Bickel and Rosenblatt \cite{Bickel}. In the nonlinear parametric autoregressive framework, Liebscher \cite{Liebscher} obtained almost sure uniform convergence of the KDE on compact sets and asymptotic normality results. Convergence rates are improved by M\"uller et al. \cite{Muller.etal} with the use of weighted kernel density estimators. Moreover, Cheng extended in \cite{Cheng2010} his results of \cite{Cheng05} in the nonlinear case. A goodness of fit test of the errors was also derived in Lee and Na \cite{Lee-Na} and Bachmann and Dette \cite{Bachmann-Dette}. Conditions on the stationarity of the time-series are given in all these references.

The nonparametric framework has been poorly addressed up to now. It only concerns the regression case. It was first studied by Ahmad \cite{Ahmad} in a fixed design regression model. He proved pointwise and uniform almost sure convergence of the noise KDE, but without providing convergence rates. In a more general regression setting, Cheng \cite{Cheng04} gave sufficient conditions under which the density estimator based on nonparametric residuals is consistent. One of these conditions is that the estimation error of the nonlinear regression function has to be uniformly weakly consistent. Efromovich \cite{Efromovich} pointed out that the nonparametric framework for error density estimation is ``extremely complicated due to its indirect nature''. He made developments under the customary assumption that the regression function is differentiable and the error density is twice differentiable. More recently, Plancade \cite{Plancade} proposed a density estimator constructed by model selection and applied it in the nonparametric regression framework. 

\bigskip

In this paper we are interested in estimating the error density function of a functional autoregressive models of order~$1$. This framework combines the difficulties encountered both in the nonparametric regression setting and in the autoregressive setting. 
Models have the following general form
\begin{eqnarray}
\label{Model}
X_{n}\ =\ f(X_{n-1})\ +\ \veps_{n} \hskip 0.5cm (n \in \bkN),
\end{eqnarray}
where $X_n \in \bkR^d$ is observed, the function $f$ of $\bkR^d$ in $\bkR^d$ is unknown and $\veps = \pa{\veps_n}_{n\geq 0}$ is the driven noise with zero mean, positive definite covariance matrix $\Gamma$ and unknown probability density function $p$. The initial state $X_0$ is given and is independent of $\veps$. 

Since the white noise $(\varepsilon_n)_{n\geq 1}$ is not observed, 
we have to construct a predictor sequence $(\widehat\varepsilon_n)_{n\geq 1}$. 
If $f$ was known, $\tilde\varepsilon_n = X_n - f(X_{n-1})$ would be a good predictor of $\varepsilon_n$. 
However, since $f$ is unknown, we have to estimate it 
in such a way that the residual $\widehat\varepsilon_n = X_n - \widehat f_{n-1}(X_{n-1})$ is a ``good'' predictor of $\varepsilon_n$, 
where $\widehat f_{n}$ is an estimator of $f$.
The case of functional autoregressive models provides an upper difficulty for the analysis of residuals. 
The objective of the present paper is to propose a residual-based recursive kernel estimator for the innovation density 
in that case, and to study its asymptotic properties. 

\bigskip

To estimate the unknown pdf $p$ we use a recursive version of 
the well-known Parzen-Rosenblatt kernel-based density estimator: for any $y\in\bkR^d$,
we estimate $p(y)$ by  
\begin{eqnarray}\label{defpchap}
\pchap (y) & = & \frac{1}{n} \sum_{i=1}^{n}\
i^{\alpha d} K\paB{i^\alpha\pab{X_i - \widehat f_{i-1}(X_{i-1}) -y}},
\end{eqnarray}
where $K$ is a kernel function and the bandwidth parameter $\alpha$ is 
a real number in $]0,1/d[$. 
The choice of a recursive estimator was favored to allow the use of martingale techniques
in exploring the asymptotic properties of $\pchap$. Recursive estimators have also the advantage of not requiring 
the stationarity of $(X_n)$ from the initial instant.

Under adapted regularity conditions on the density function $p$, the link between the estimation error of $\pchap$ and the errors of $\fchap$ may
be formulated as follows: for all $y \in \bkR^d$, 
\begin{equation*}
 | \widehat{p}_n(y) -p(y)| 
= O\left(\frac{1}{n}\sum_{i=0}^{n-1}{|| \widehat{f}_i(X_i) - f(X_i) ||}\right)  + o (1)\quad \mbox{a.s.}
\end{equation*}
That is,  the estimation error of $\pchap$ will always depend on the average prediction error of $\fchap$, 
whose convergence to $0$, ie.
\begin{equation}
\label{res-error-pred}
\frac{1}{n}\sum_{i=0}^{n-1}{|| \widehat{f}_i(X_i) - f(X_i) ||} = o(1) \quad \mbox{a.s.}
\end{equation}
is the major difficulty in proving the convergence of $\pchap$ to $p$.
It is clear that since the process $(X_n)$ is not bounded, this last result requires strong convergence results on $\fchap$. 
The main difficulty is then to find an estimator of $f$ that meets this requirement. 
Since no structural assumption is set on $f$, we choose to use a recursive version of the well-known 
Nadaraya-Watson kernel estimator \cite{Nadaraya,Watson}, 
studied for example by Senoussi \cite{Senoussi}, see also Duflo \cite{Duflo}. Proving (\ref{res-error-pred}) with this estimator will be the first step to achieve before studying the asymptotic properties of $\pchap$.

\bigskip

The paper is organized as follows.
The framework and the assumptions are presented in Section 2, together with the properties they induce on model (\ref{Model}).
Section 3 is dedicated to the study of the nonparametric kernel estimator of $f$: strong consistency and conditions for achieving
(\ref{res-error-pred}). Asymptotic properties of the KDE $\pchap$ (\ref{defpchap}) are studied in Section 4: 
uniform strong consistency and central limit theorem (CLT for short).
Proofs of the main results are postponed in appendix.


\section{Model assumptions and properties}
\setcounter{equation}{0} 

The following set of assumptions is common when dealing with autoregressive 
functional models (Duflo\,\cite{Duflo}). 
\vspace{0.3cm}

\noindent
{\sc Assumption {\rm [A1]}}.\ \ 
{\it Function $f$ is continuous and there are two positive constants\
$r_f < 1$ and $C_f$ such that for any $x\in \bkR^d$,}
\begin{equation} \label{MajNorm_f}
\norm{f(x)}\ \leq\ r_f \norm{x} + C_f.
\end{equation}

\noindent
{\sc Assumption {\rm [A2]}}.\ \ 
{\it The initial state $X_0$ and $\veps = \pa{\veps_n}_{n \geq 0}$ have a finite moment
of order $m > 2$.}
\vspace{1ex}

\noindent
These assumptions will ensure good stability properties of the process $(X_n)_{n \ge 0}$.
In particular, since by [A2] the noise $(\veps_n)$ has a finite moment of order $m>2$,
then $\displaystyle\varepsilon^{\sharp}_n:=\sup_{i \le n} \| \varepsilon_i \| = o\pa{n^{1/m}}$ a.s. 
and we derive from Proposition 6.2.14 of Duflo\,\cite{Duflo} that almost surely
\begin{equation} \label{MajSumetSupNormX}
\sum_{i=1}^{n} \norm{X_i}^m  = O\pa{n}
\quad\quad\mbox{and}\quad\quad
\sup_{i\leq n}\norm{X_i} = O\pa{\veps_n^\#} = o\pa{n^{1/m}}
\end{equation}
These two results will be useful in the rest of the paper. 

\subsection{Strengthening Assumption [A2]}

Assumption [A2] is rather standard and holds for many probability distributions. However it may 
be interesting to restrict studies to particular subfamilies of noises depending on their tail 
distribution. It is particularly useful to get more precise 
properties, as for example better convergence results or better asymptotic bounds.
Restrictions to noises with a finite exponential moment and to Gaussian noises are presented 
in this paragraph. 
\vspace{1ex}

\noindent
{\sc Assumption {\rm [A2bis]}}.\ \ 
{\it There is $m>0$ such that $\bkE\cro{\exp\pa{m\norm{X_0}}} < \infty$ and
$\bkE\cro{\exp\pa{m\norm{\veps_1}}} < \infty$.}

\vspace{1ex}

\noindent
Conditions of Proposition 6.2.15 of \cite{Duflo} are verified with 
(\ref{MajNorm_f}) and [A2bis], which implies that, for any $a < m$, almost surely
\begin{equation}\label{MajSumExpNormX}
\sum_{i=1}^{n} \exp(a\norm{X_i})\ =\ O(n)
\quad\quad\mbox{and}\quad\quad
\sup_{i\leq n}\norm{X_i} = o\pa{\log n}.
\end{equation}

\vspace{1ex}

In the same spirit of [A2bis], 
we shall be interested on what happens with Gaussian white noises. This is the subject of 
the next Proposition, which is an adaptation of Proposition 6.2.15 of \cite{Duflo}.

\begin{proposition}\label{PropMajGauss}
Consider Model (\ref{Model}) where $(\veps_n)$ is a Gaussian white noise with
invertible covariance matrix $\Gamma$. This implies that
there is $m < 1 / 2\lambda_{\min}\pa{\Gamma}$ such that\ $\bkE\cro{\exp(m\norm{\veps_1}^2)} < \infty$.
Assume that $\bkE\cro{\exp(\norm{X_0}^2 / 2\lambda_{\min}\pa{\Gamma})} < \infty$.
Assume also that $f$ is continuous and that there is $c_f\in]0\,,\,1[$ such that
\begin{equation}\label{HypBisSur_f}
\liminf_{\|x\|\rightarrow \infty}\pa{c_f \norm{x}^2 - \norm{f(x)}^2} > \frac{c_f}{m(1 - c_f)} \log\pa{\bkE\cro{\exp(m\norm{\veps_1}^2)}}.
\end{equation}
Then,  
\begin{equation}\label{MajSupNormXGauss}
\sup_{i\leq n}\norm{X_i}\ =\ o\pa{\sqrt{\log n}}\quad\mbox{a.s.}
\end{equation}
and for any $a < (1-c_f) / 2\lambda_{\min}\pa{\Gamma}$,
\begin{equation}\label{MajSumExpX2}
\sum_{i=1}^{n} \exp(a\norm{X_i}^2)\ =\ O(n)\quad\mbox{a.s.}
\end{equation}
\end{proposition}

\begin{proof}
Under assumption (\ref{HypBisSur_f}), for some finite constants $M > 0$
and $b > 0$, if $\norm{x} > M$, we have
\begin{equation}\label{Majf}
\norm{f(x)}^2\ \leq\ c_f \norm{x}^2 -  \frac{c_f}{m(1-c_f)} \log\pa{\bkE\cro{\exp(m\norm{\veps_1}^2)}} - b
\end{equation}
In addition as $f$ is continuous, then $\sup_{\norm{x}\leq M}\norm{f(x)} < \infty$.
Let us set $Z_n = \exp\pa{m(1 - c_f)\norm{X_n}^2}$
and ${\cal F}_n = \sigma\pa{X_0, \veps_1, \cdots, \veps_n}$.
For any $0 < c < 1$ and any $x,y\in\bkR^d$, we have $\norm{x + y}^2 \leq \norm{x}^2 / c  + \norm{y}^2 / (1 - c)$.
Thus,
\begin{equation}
\bkE\cro{Z_{n+1}\,|\,{\cal F}_n}\ \leq\ 
\bkE\cro{\exp(m\norm{\veps_1}^2)}\,
\exp\pa{\frac{m(1 - c_f)}{c_f}\norm{f(X_n)}^2}
\end{equation}
and using (\ref{Majf}), we derive that for some positive constants $c$ and $C$,
\begin{equation}
\bkE\cro{Z_{n+1}\,|\,{\cal F}_n}\ \leq\ e^{-c}\,Z_n\ +\ C
\end{equation}
Finally, applying Proposition 6.2.12 of \cite{Duflo} with the Lyapounov function 
$V(x) = \exp\pa{m(1 - c_f)\norm{x}^2}$, we obtain (\ref{MajSupNormXGauss}) and (\ref{MajSumExpX2}) for any $a < m(1 - c_f)$
and any $m < 1 / 2\lambda_{\min}\pa{\Gamma}$.
This closes the proof of Proposition \ref{PropMajGauss}.
\end{proof}

\begin{remark}\label{RemBruitGaus}
{\rm When $f$ satisfies [A1], we can find $c_f\!\in] r_f\,,\,1[$ such 
that assumption (\ref{HypBisSur_f}) is fulfilled.
Indeed (\ref{MajNorm_f}) implies that $\norm{f(x)}^2 \leq r_f \norm{x}^2 + C_f^2 / (1 - r_f)$. 
Hence, when $f$ satisfies [A1] and $(\veps_n)$ is Gaussian, then  
(\ref{MajSupNormXGauss}) and (\ref{MajSumExpX2}) holds for any $a \leq (1 - r_f) / 2\lambda_{\min}\pa{\Gamma}$.
In particular, we have
\begin{equation}\label{MajSumExpNormX2}
\sum_{i=1}^{n} \exp\pab{(1 - r_f)\norm{X_i}^2/ 2\lambda_{\min}\pa{\Gamma}}\ =\ O(n)
\quad\mbox{a.s.}
\end{equation}
}\end{remark}

This stability property will be useful when dealing with Gaussian noises.

\subsection{Asymptotic stationarity and properties}

A main consequence of this framework  is a property of stationarity. 
Indeed, with [A1] and by assuming that the distribution of 
$\pa{\veps_n}$ has a probability density function $p > 0$, the process $X = \pa{X_n}_{n \geq 0}$
is asymptotically stationary and possesses an invariant distribution $\mu$
which has a finite moment of order $m$ and a 
probability density function denoted\ $h$, which satisfies for any $x\in\bkRd$\,:
\begin{eqnarray}\label{Proph}
h(x) & = & \int_\bkRd\! p\pab{x - f(t)}\,h(t)\,\dt
\end{eqnarray}

Moreover, the following property holds:\
for any $\mu$-integrable function\ $g~: \bkR^d \rightarrow \bkR$ which satisfies
$\| g(x) \| \le C ( \left\|x\right\|^m +1)$ (where $C$ is a constant),  
the strong law of large numbers states that
\begin{eqnarray}\label{LGNg}
\frac{1}{n}\sum_{i=0}^{n-1}~g\pa{X_i} & \tendas & \int_\bkRd\! g(x)~{\rm d}\mu(x) .
\end{eqnarray}
Besides, for a positive constant $R$, we have
\begin{eqnarray*}
\frac{1}{n} \sum_{i=0}^{n-1}{\ind_{\acc{\|f(X_i)\| < R}}}& \geq & 1 - \frac{1}{n R} \sum_{i=0}^{n-1} \norm{f(X_i)}.
\end{eqnarray*}
Thus, applying (\ref{LGNg}) yields that, for  $R > \int_{\bkR^d} \left\|f(x)\right\| h(x) {\rm d}x$,
\begin{eqnarray}
\label{liminf-ball}
\liminf_{n \to \infty} \frac{1}{n} \sum_{i=0}^{n-1}{\ind_{\acc{\| f(X_i) \| < R}}} & \geq & 1 - \frac{\tau}{R}\ >\ 0\quad\mbox{a.s.} \; ,
\end{eqnarray}
which means that the process $(X_n)_{n \geq 0}$ infinitely often crosses the ball of radius $R$ centered on $0$. This last property will be useful for proving convergence results of $\fchap$ over dilated sets. 

\section{Properties of the kernel estimator $\fchap$}
\label{section-fchap}
\setcounter{equation}{0}

We shall now introduce the estimator of $f$. 
Since no structural assumption was set on function $f$, 
we chose a recursive nonparametric estimator, following the well-known Nadaraya-Watson estimator.
Let $K$ be a kernel and $\beta$ a real number in $]0, 1/ d[$.
Then, for any $x\in\bkR^d$, we estimate $f(x)$ by 
\begin{eqnarray}
\fchap (x) & = & \frac{\sum_{i=1}^{n-1}\,i^{\beta d} K\pab{i^\beta\pa{X_i -x}} X_{i+1}}
{\sum_{i=1}^{n-1}\,i^{\beta d} K\pab{i^\beta\pa{X_i -x}}}
\label{deffchap}
\end{eqnarray}
if the denominator of (\ref{deffchap}) is not equal to 0,
and by 0, otherwise.
As $\pchap$, it has a recursive form which allows the use of martingale properties. For notation ease, $\fchap$ is defined with the same 
kernel function $K$. It is of course possible to take another one, provided that it has the same characteristics.
We also point out that the denominator of (\ref{deffchap}), when divided by $n$, 
is an estimate of $h(x)$, the stationary distribution density.
\vspace{1ex}

Beyond assumptions [A1] and [A2] on $f$ and $(\veps_n)$, 
we impose the following properties to function $f$ and the pdf $p$:
\vspace{1ex}

\noindent
{\sc Assumptions {\rm [A3]}}.\ 
{\sl Function $f$ belongs to $C^1(\bkRd)$ and
its first derivatives are bounded.}
\vspace{1ex}

\noindent
{\sc Assumptions {\rm [A4]}}.\ 
{\sl The probability density function $p$ is positive and belongs to $C^1(\bkRd)$, and
$p$ and its first derivatives are bounded.}
\vspace{1ex}

Furthermore, all the proofs of the paper are based on a kernel function $K$ 
with the  characteristics given hereafter: 
\vspace{1ex}

\noindent
{\sc Assumptions {\rm [A5]}}.\ 
{\sl The kernel $K$ is a nonnegative function, Lipschitz, bounded with compact support, and integrates to 1.}

\subsection{Strong uniform consistency of $\fchap$}

Under Assumptions [A1]--[A5], Duflo \cite{Duflo} and Senoussi \cite{Senoussi} prove the almost sure pointwise
convergence of $\fchap$ to $f$, as well as a pointwise central limit theorem and results 
of uniform convergence on compact sets.
Since the process $(X_n)_{n \ge 0}$ is not bounded, these convergence results are not sufficient to get (\ref{res-error-pred}). 
The uniform convergence over dilated sets of $\fchap$ is necessary. 
This kind of convergence has already been established in the regression framework by Bosq \cite{Bosq}
under mixing conditions.
For model (\ref{Model}) and when $\veps$ is Gaussian, Duflo \cite{Duflo} also proves that
for $c$ sufficiently small there is $s > 0$ such that
\begin{equation*}
\sup_{\left\|x\right\| \le c\sqrt{\log n}} \left\|\fchap(x)-f(x)\right\| = o\pa{n^{-s}} 
\quad \mbox{\rm a.s.}
\end{equation*}
In a control framework, that is, when $X_n$ is submitted to the action of an exogenous variable (also called control variable),
Portier and Oulidi \cite{Portier-Oulidi} establish the same kind of result but with a more general noise. 
We will now adapt these convergence results to model (\ref{Model}) and also improve them to study the prediction errors, which has never been done before. 
For the sequel, let us denote 
$$m_n = \inf\Bigl\{p(z)\,;\,\norm{z} \leq v_n+R\Bigr\}$$
where $\pa{v_n}_{n\geq 0}$ is a sequence of positive real numbers
increasing to infinity and $R$ is a constant greater than $ \int_{\bkR^d} \left\|f(x)\right\| h(x) {\rm d}x$.
\vspace{1ex}

An additional assumption on the probability density function $p$ must be introduced to study the denominator of $\fchap$.
\vspace{1ex}

\noindent
{\sc Assumption {\rm [A6]}}.\ \ \ {\sl There is a sequence of
positive real numbers\ $\pa{v_n}_{n\geq 1}$ increasing to infinity such that
\ $v_n = O\pa{n^\nu}$ with $\nu > 0$\  and 
\[m_n^{-1}\ =\ \inf\Bigl\{o\pa{n^\beta}\,,\,O\pa{n^{1 - s}}\Bigr\}\]
where\ $s\in](1+\beta d)/2\,,\,1[$ and $\beta\in]0\,,\,1/d[$}.
\vspace{1ex}

Let us mention that [A6] is not required to establish pointwise or uniform on compact sets
convergence results for $\fchap$. 
Besides, when $\beta < 1/(d+2)$, then [A6] reduces to $m_n^{-1} = o(n^\beta)$.
Indeed, in that case, we can find $s\in](1+\beta d)/2\,,\,1[$ such that $\beta = 1 - s$.

\begin{theorem}\label{cvud.fchap}
Let $\beta\in ]0,1/2d[$.
Assume that [A1]--[A6] hold. 
Then, we have
\begin{equation}\label{majsupfchap}
\sup_{\left\|x\right\| \le v_n} \left\|\fchap(x)-f(x)\right\| = o\left(\frac{n^{\lambda-1}}{m_n}\right) + O\left(\frac{n^{-\beta}}{m_n}
\right) \quad \mbox{\rm a.s.}
\end{equation}
for any $\lambda\in]\frac{1}{2} + \beta d\,,\,1[$.
In particular, for any $\beta < 1/ 2(d+1)$,
\begin{equation}\label{maj2supfchap}
\sup_{\left\|x\right\| \le v_n} \left\|\fchap(x)-f(x)\right\|\ =\ O\left(\frac{n^{-\beta}}{m_n}
\right) \quad \mbox{\rm a.s.}
\end{equation}
\end{theorem}

\begin{proof}
The proof is postponed in Appendix B.
\end{proof}

The bound given in (\ref{maj2supfchap}) shows that the convergence rate of $ \fchap$ over dilated sets 
strongly depends on the decrease of the density function $p$ 
and on the choice of a well-suited sequence $(v_n)$. 
In particular, the way the pdf $p$ decreases at infinity has to be known to choose an 
appropriate sequence $(v_n)$. 
For example, in the Gaussian case, it is easy to see that for any $c\in]0\,,\,1[$,
\begin{equation}\label{MinDensGauss}
m_n \geq \mbox{const.}\exp\pa{-v_n^2\,/\, 2 c \lambda_{\min}\pa{\Gamma}}
\end{equation}
Therefore, choosing\ $v_n = A\pa{\log\log\,n}^{1/2}$\ with $A > 0$,
we obtain that $m_n^{-1} = O((\log n)^{A^2 \,/\,2c\lambda_{\min}\pa{\Gamma}})$.
Assumption [A6] is then satisfied and 
we derive from (\ref{majsupfchap}) that for $\beta = 1/ 2(d+1)$ and any $\lambda\in]0\,,\,1/ 2(d+1)[$, 
\begin{equation}\label{maj2supfchapgauss}
\sup_{\left\|x\right\|^2 \le A \log\!\log n} \left\|\fchap(x)-f(x)\right\|\ =\ o\left(n^{-\lambda}\right) \quad \mbox{\rm a.s.}
\end{equation}

\subsection{Average prediction  error of $\fchap$}
\label{AvPredErr.fchap}

As mentioned in introduction, we are in fact interested in the asymptotic average prediction error 
of $\fchap$. It is studied in the next Corollary.

\begin{corollary}\label{CorErrPred}
Assume that [A1]--[A6] hold. 
Assume also that the sequence $(m_n)$ is decreasing.
Then, for any $\beta < 1/ 2(d+1)$,
\begin{equation}\label{MajErrPred}
\frac{1}{n} \sum_{i=1}^{n} \|\widehat{f}_i(X_i) - f(X_i) \|
\ =\ O\pa{\frac{n^{-\beta}}{m_n}} +  O\pa{\frac{\wunn + \wdeuxn}{n}}\quad \mbox{\rm a.s.}
\end{equation}
where\ $\wunn = \sum_{i=1}^{n} v_i^{1-m}$\ and\ 
$\wdeuxn = \sum_{i=1}^{n} i^{1/m}\,v_i^{-m}$ if the sequence $(n^{1/m}\,v_n^{-m})$
is decreasing and $\wdeuxn =  n^{1/m}\sum_{i=1}^{n} v_i^{-m}$ otherwise.

\noindent
Moreover, if [A2bis] holds instead of  [A2], then 
$\wunn = \sum_{i=1}^{n} v_i \exp\pa{-a v_i}$ for any $a < m$\ and\ 
$\wdeuxn = \sum_{i=1}^{n} (\log i)\exp\pa{-a v_i}$ if the sequence $\pab{(\log n)e^{-a v_n}}$
is decreasing and $\wdeuxn =  (\log n)\sum_{i=1}^{n} \exp\pa{-a v_i}$ otherwise.
\end{corollary}

\begin{proof}
For any $n$, let us denote by $\pi_n$ the prediction error 
defined by $\pi_n = \fchap(X_n) - f(X_n)$.
To establish (\ref{MajErrPred}), we consider the following decomposition
\begin{equation}  \label{decomp-pred-error-fchap}
\frac{1}{n} \sum_{i=1}^{n} \norm{\pi_i} = 
\frac{1}{n} \sum_{i=1}^{n} \norm{\pi_i} \indEi
+ \frac{1}{n} \sum_{i=1}^{n} \norm{\pi_i} \indEic
\end{equation}
On one hand, we easily deduce that 
\begin{equation}\label{aver-pred-error-term1}
\frac{1}{n} \sum_{i=1}^{n} \norm{\pi_i} \indEi
\leq \frac{1}{n} \sum_{i=1}^{n} \sup_{\|x\| \leq v_i} \norm{\widehat f_i(x) - f(x)}
\end{equation}
Now, using (\ref{maj2supfchap}) and the fact that $(m_n)$ 
is assumed to be decreasing, we derive that
\begin{equation}\label{ResErrPred1}
\frac{1}{n}\sum_{i=0}^{n-1} \| \widehat{f}_i(X_i) - f(X_i) \| \indEi  
= O\pab{n^{-\beta} m_n^{-1}}\quad\mbox{a.s.} 
\end{equation}
On the other hand, thanks to (\ref{majfchap}) in Appendix B, we infer that
\begin{equation}\label{MajErrPred2}
\sum_{i=1}^{n} \norm{\pi_i} \ind_{\acc{\left\|X_i\right\| > v_i}}
 \leq  \sum_{i=1}^{n} \pa{\norm{X_i} + \varepsilon^{\sharp}_i + C_{f,K}}\ind_{\acc{\left\|X_i\right\| > v_i}} \quad \mbox{a.s.}
\end{equation}
The bound of (\ref{MajErrPred2})
completely depends on the moment assumption on $(\veps_n)$.
Indeed, when $(\veps_n)$ satisfies [A2], we have $\varepsilon^{\sharp}_n = o(n^{1/m})$ and 
we know by (\ref{MajSumetSupNormX}) that 
$\sum_{i=1}^{n} \norm{X_i}^m = O(n)$ a.s.
Therefore, by Lemma\,\ref{LemmaMajSomInd} with $Z_n = \norm{X_n}$
and $g(z) = z^m$, we derive that
\begin{equation}
\sum_{i=1}^{n} \pa{\norm{X_i} + C_{f,K}}\ind_{\acc{\left\|X_i\right\| > v_i}} 
= O\pab{\sum_{i=1}^{n} v_i^{1- m}} \quad \mbox{a.s.}
\end{equation}
which defines the term $\wunn$.
In addition, if the sequence $\pab{n^{1/m}\,v_n^{-m}}$ is decreasing, 
according to Lemma\,\ref{LemmaMajSomInd}
we have 
\begin{equation} \label{sum-sup-noise}
\sum_{i=1}^{n} \varepsilon^{\sharp}_i\, \ind_{\acc{\left\|X_i\right\| > v_i}} 
= O\pab{\sum_{i=1}^{n} i^{1/m}\, v_i^{- m}} \quad \mbox{a.s.}
\end{equation}
which defines the term $\wdeuxn$.
If the sequence $\pab{n^{1/m}\,v_n^{-m}}$ is not decreasing, we obtain
\[
\sum_{i=1}^{n} \varepsilon^{\sharp}_i\, \ind_{\acc{\left\|X_i\right\| > v_i}} 
\leq \varepsilon^{\sharp}_n \sum_{i=1}^{n} \ind_{\acc{\left\|X_i\right\| > v_i}} 
= o\pab{n^{1/m} \sum_{i=1}^{n} v_i^{- m}} \quad \mbox{a.s.}
\]
which defines the second form of $\wdeuxn$.
Hence, combining the previous results, we obtain
\begin{equation}
\frac{1}{n} \sum_{i=1}^{n} \| \widehat{f}_i(X_i) - f(X_i) \| \ind_{\acc{\left\|X_i\right\| > v_i}}
 =  O\pa{\frac{\wunn + \wdeuxn}{n}} \quad \mbox{a.s.}
\label{ResErrPred2}
\end{equation}
and the first part of Corollary\,\ref{CorErrPred} is established combining (\ref{ResErrPred1}) and
(\ref{ResErrPred2}).

\noindent
Now, if $(\veps_n)$ has a finite exponential moment of order $m$, 
then $\varepsilon^{\sharp}_n = o(\log n)$ a.s. and we know by 
(\ref{MajSumExpNormX}) that for any $a<m$, $\sum_{i=1}^{n} \exp\pa{a \norm{X_i}} = O(n)$ a.s.
Therefore, proceeding in the same manner, we find using Lemma\,\ref{LemmaMajSomInd} that
for any $a < m$,\ $\wunn = \sum_{i=1}^{n} v_i \, e^{- a v_i}$ \ and \
$\wdeuxn = \sum_{i=1}^{n} (\log i) e^{- a v_i}$ if the sequence
$\pab{( \log n) e^{- a v_n}}$ is decreasing and\ 
$\wdeuxn = \log n\sum_{i=1}^{n} e^{- a v_i}$ otherwise.
This achieves the proof of Corollary\,\ref{CorErrPred}.
\end{proof}

Corollary\,\ref{CorErrPred} does not state that the average prediction error of $\fchap$ converges to
0, see (\ref{res-error-pred}).
Obtaining such a result depends on the decrease at infinity of $p$ and
the existence of a well-suited sequence $(v_n)$.
The moment assumption on $(\veps_n)$ gives a first information on 
the possible choices of $(v_n)$.
As $(v_n)$ increases to infinity, we always have $\wunn / n = o(1)$.
Therefore, the sequence $(v_n)$ must be chosen in such a way that 
$\wdeuxn / n = o(1)$, which depends on the moment assumption.
Moreover, dealing with the term $n^{-\beta} / m_n$ requires to know the way $p$ is decreasing. 
The sequence $(v_n)$ has thus to be selected from a compromise between these two conditions. 
We give below three examples that well illustrate this compromise.

\begin{example} \label{Example-dens-dec-pol}
{\rm
Assume that the decrease at infinity of $p$ is of the form
$C \norm{x}^{-\delta}$ with $C > 0$ and $\delta > 3$.
Then, $(\veps_n)$ has a finite polynomial moment of order $m$ with $m\in]2,\delta -1[$.
Let us choose the sequence $(v_n)$ under the form $v_n = A n^{\eta}$ for 
some positive constants $A$ and $\eta < \beta /  \delta$. 
It follows that
\begin{equation*}
m_n^{-1} = O\left( (v_n+R)^{\delta}\right) = O \left( n^{\delta \eta}  \right)
\end{equation*}
and [A6] holds since $\eta < \beta / \delta$.
In addition, as $m > 2$, we have $\wunn = o(n)$ and as soon as $\eta > 1 / m^2$,
we also have $\wdeuxn = o(n)$.
Of course, the condition 
$1 / m^2 < \eta < \beta /  \delta$ with $\beta < 1 / 2(d+1)$
implies that $\delta$ must be sufficiently large to ensure that 
$\delta / m^2 < 1/2(d+1)$. 
Thus, for such $\delta$ and any $\beta \in ]\delta/m^2\,,\,1/2(d+1)[$, 
the prediction errors satisfy (\ref{res-error-pred}) and, since $\eta < \beta/\delta < 1 / m$,
we have $\wunn = O(\wdeuxn)$ and therefore
\begin{equation*}
\frac{1}{n} \sum_{i=1}^{n} \| \widehat{f}_i(X_i) - f(X_i) \|  =  O\pa{n^{-\beta + \eta \delta}}
+ O\pa{n^{-m\eta + 1/m}}\quad \mbox{a.s.}
\end{equation*}
The best rate of convergence is obtained by taking  
$\eta = (\beta + 1/m) /(m+\delta)$ and we find that for any 
$\beta \in ]\delta/m^2\,,\,1/2(d+1)[$,
\begin{equation}\label{BestBoundErrPred}
\frac{1}{n} \sum_{i=1}^{n} \| \widehat{f}_i(X_i) - f(X_i) \|  =  O\pa{n^{-\beta \tau}}\quad \mbox{a.s.}
\end{equation}
with $\tau = (m - \delta/m\beta) /(m+\delta)$.\ $\emptysq$
}\end{example}

\begin{example} \label{Example-dens-dec-exp}
{\rm 
Assume the decrease at infinity of $p$ is of the form
$C  \exp\pa{-\delta\norm{x}}$ with $C > 0$ and $\delta > 0$.
Thus, $(\veps_n)$ has a finite exponential moment of order $m < \delta$.
Let us choose the sequence $(v_n)$ under the form $v_n = \eta \log n$ 
with $\eta >0$. 
Then, $\wunn + \wdeuxn = O\pa{n^{1- a \eta} \log n}$ with $a < m$, 
$m_n^{-1} = O \left( n^{\delta \eta} \right)$ and 
[A6] holds as soon as $\eta < \beta / \delta$.
In that case, the prediction errors satisfy (\ref{res-error-pred}).
The best bound is obtained by choosing $\eta = \beta / 2m$ :
for $\beta < 1 /2(d+1)$, result (\ref{BestBoundErrPred}) stands
for any $\tau < 1/2$.\ $\emptysq$
}\end{example}

\begin{example} \label{Example-dens-gauss}
{\rm 
Assume now that $(\veps_n)$ is Gaussian with covariance matrix $\Gamma$. 
Thanks to Remark \ref{RemBruitGaus}, we know
that\ $\sum_{i=1}^{n} \exp\pab{ a \norm{X_i}^2}= O(n)$\ 
where $a = (1 - r_f) 2\lambda_{\min}\pa{\Gamma}$, and
$\veps_n^\# = o\pa{\sqrt{\log n}}$ a.s.
Hence following the proof of Corollary\,\ref{CorErrPred}, we find
that result (\ref{MajErrPred}) holds with 
$\wunn = \sum_{i=1}^{n} v_i \, e^{- a v_i^2}$ \ and \
$\wdeuxn = \sqrt{\log n}\sum_{i=1}^{n} e^{- a v_i^2}$.
Therefore, taking $v_n = (\eta\,\log n)^{1/2}$ with $\eta > 0$
leads to $\wunn + \wdeuxn = O\pa{\sqrt{\log n}\,n^{1 - a \eta}} = o(n)$.
In addition, using (\ref{MinDensGauss}), we derive
that $m_n^{-1} = O(n^{\eta / 2 c \lambda_{\min}(\Gamma)})$ for any $c \in ]0,1[$.
Then [A6] is fulfilled by choosing $\eta < 2\beta c \lambda_{\min}(\Gamma)$
and the prediction errors satisfy (\ref{res-error-pred}).
Moreover, choosing $\eta = 2\beta c \lambda_{\min}(\Gamma) / (1 + c(1 - r_f))$ implies the better rate of convergence:
for any $\beta < 1/2(d+1)$ and any $c\in]0,1[$, 
result (\ref{BestBoundErrPred}) stands
for $\tau  = c(1 - r_f) / (1 + c(1 - r_f) < 1/2$.\ $\emptysq$
}\end{example}

\section{Asymptotic properties of the KDE $\pchap$}
\setcounter{equation}{0} 

This section is concerned with the asymptotic properties
of the KDE $\pchap$.
More precisely, we present the uniform strong consistency with rate for $\pchap$
as well as a pointwise and multivariate central limit theorem. 

\subsection{Strong consistency}\label{subsec-strong-cons-pchap}

The following Theorem \ref{th.sur.p} gives the conditions 
that ensure the uniform strong consistency of $\pchap$. 

\begin{theorem}\label{th.sur.p}
Assume that [A1]--[A5] hold true. If there is a sequence $(v_n)$ of the form $v_n = An^\eta$ with $A > 0$ 
and $\eta > 1/m^2$ such that the sequence $(m_n)$ is decreasing and satisfies
$m_n^{-1} = o(n^\beta)$ where $\beta  < 1/ 2(d+1)$, then almost surely
$\sup_{y \in \bkRd } | \widehat{p}_n(y) -p(y)| = o(1)$\ 
and more precisely,
\begin{equation}\label{borneconvpchapsurRd}
\sup_{y \in \bkRd } | \widehat{p}_n(y) -p(y)| = o(n^{\gamma - 1})
+ O(n^{-\alpha}) + O\pa{\frac{n^{-\beta}}{m_n}} + O\pa{w_n}
\end{equation}
where $w_n = n^{-m\eta + 1/m}$ if $\eta \leq 1/m$ and $w_n = n^{-\eta (m-1)}$ otherwise.
Moreover, if [A2bis] holds instead of  [A2] and if there is a sequence $(v_n)$ 
of the form $v_n = \eta \log n$\ with $\eta > 0$ such that the sequence 
$(m_n)$ is decreasing and satisfies $m_n^{-1} = o(n^\beta)$ where $\beta  < 1/ 2(d+1)$, 
then $\pchap$ satisfies (\ref{borneconvpchapsurRd}) with $w_n = n^{- a \eta} (\log n)$ and $a < m$. 
\end{theorem}

\begin{proof}
The proof is postponed in Appendix C.
\end{proof}

Theorem \ref{th.sur.p} directly establishes the convergence of the estimation error of $p$ uniformly on whole $\bkR^d$, 
that is, without first studying what happens on fixed compact sets nor at fixed $y$. The reason is that no matter what set of 
$y$ you consider, you always have to verify the assumptions required for the convergence (\ref{res-error-pred}) of the 
average prediction error of $\fchap$, see Section \ref{section-fchap}. These assumptions are strong enough to directly 
obtain uniform strong consistency properties of $\pchap$ without additional hypothesis.

\vspace{1ex}

The convergence of $\pchap$ thus relies on the choice of sequence $(v_n)$ which depends on the way $p$ is decreasing at infinity. 
Then the question is: ``Which pdf are consistently estimated with $\pchap$ ?'' 
We have exhibited three families of densities (see examples in \S \ref{AvPredErr.fchap}):
densities with polynomial or exponential decrease at infinity
and Gaussian densities.
Let us specify the different results for these three examples.

\finl
{\bf Example \ref{Example-dens-dec-pol} (continued)}
Recall that the best rate of convergence for the average prediction errors
is obtained by choosing the sequence $(v_n)$ under the form $v_n = A n^{\eta}$ 
with $A > 0$, $\eta = (\beta + 1/m) /(m+\delta)$ and $\beta \in ]\delta/m^2\,,\,1/2(d+1)[$.
With this choice, we obtain : 
\begin{equation}\label{convpchapsurRdcaspoly}
\sup_{y \in \bkRd } | \widehat{p}_n(y) -p(y)| = o(n^{\gamma - 1})
+ O(n^{-\alpha}) + O\pa{n^{-\beta \tau}}
\end{equation}
with $\tau = (m - \delta/m\beta) /(m+\delta)$.\ $\emptysq$

\vspace{1ex}

The following remark shows how the convergence rate can be improved and constraint on $\delta$
relaxed in the polynomial case,  by considering an another KDE of $p$.

\begin{remark}\label{RemEstTronque}
{\rm The main difficulty in proving the consistency of $\pchap$ comes from the study of the prediction errors
$\fchap(X_n) - f(X_n)$, and in particular from the term
$\sum_{i=1}^{n} \norm{\widehat f_i(X_i) - f(X_i)} \ind_{\acc{\left\|X_i\right\| > v_i}}$
(see Corollary \ref{CorErrPred}).
This term is studied using a crude upper bound involving $\veps_n^\#$.
In the framework of [A2], this upper bound led us to introduce additional and 
restrictive constraints on $\eta$ and $\beta$ for proving that $\wdeuxn = o(n)$.
One way to avoid the study of such a term is to adapt the plug-in estimator $\pchap$ 
with truncated residuals 
$\widehat\veps^\ast_n = X_n - \widehat f_{n-1}(X_{n-1})\ind_{\acc{\|X_{n-1}\| \le v_{n-1}}}$
%
%
leading to the following truncated version of the kernel density estimator:
\begin{eqnarray}\label{defpchapetoil}
\pchap^\ast (y) & = & \frac{1}{n} \sum_{i=1}^{n}\
i^{\alpha d} K\paB{i^\alpha\pab{X_i - \widehat f_{i-1}(X_{i-1})\ind_{\acc{\|X_{i-1}\| \le v_{i-1}}} -y}}
\end{eqnarray}
The sequence $(v_n)$ still has to satisfy
assumption [A6] to provide the uniform almost sure convergence over dilated sets of $\fchap$. 
In the context of [A2], we choose $v_n$ of the form $v_n = A n^\eta$
with $A > 0$ and $\eta > 0$.
Following the proof of Theorem\,\ref{th.sur.p},
we easily show that under [A1]-[A5]
\begin{equation*}
\sup_{y \in \bkRd } | \widehat{p}^\ast_n(y) -p(y)| = o(n^{\gamma - 1})
+ O(n^{-\alpha}) + O\pa{\frac{n^{-\beta}}{m_n}} + O\pa{n^{-(m-1)\eta}} \quad \mbox{\rm a.s.}
\end{equation*}
if $(v_n)$ is such that $(m_n)$ is decreasing and $m_n^{-1} = o(n^\beta)$.
In particular, in the context of Example\,\ref{Example-dens-dec-pol}, the best rate of convergence is obtained taking
$\eta = \beta / (m+\delta -1)$ with $\beta < 1/2(d+1)$.
With that choice, the KDE $\pchap^\ast$
satisfies (\ref{convpchapsurRdcaspoly}) 
with $\tau  = (m-1) / (m+\delta -1)$, which yields a better convergence rate than that of 
$\pchap$ since  $(m - 1)/(m+\delta - 1) >\!> (m - \delta/m\beta)/(m+\delta)$. This result 
has moreover been obtained by only assuming that $\delta > 3$.

\noindent
Of course, the main drawback of this KDE based on truncated residuals 
is its use in practice, as the sequence $(v_n)$ shall of course be suitably chosen.
}\end{remark}

\vspace{1ex}

\finl
{\bf Example \ref{Example-dens-dec-exp} (continued)}
When the decrease of $p$ is exponential, the KDE $\pchap$ satisfies
(\ref{convpchapsurRdcaspoly}) for any $\tau < 1/2$.\ $\emptysq$  

\finl
{\bf Example \ref{Example-dens-gauss} (continued)}
In the Gaussian case, the KDE $\pchap$ satisfies
(\ref{convpchapsurRdcaspoly}) 
with $\tau  = c(1 - r_f) / (1 + c(1 - r_f)$ for any $c\in]0,1[$.\ $\emptysq$

\subsection{Central limit theorem}

In this section, we present a pointwise and multivariate central limit theorem for $\pchap$.

\begin{theorem}\label{TheoCLTpchap}
Assume that [A1]--[A5] hold.
Assume also there is a sequence $(v_n)$ of the form $v_n = An^\eta$ with $A> 0$ 
and $\eta \in ]1/m^2\,,\,1/m^2 + 1/m(d+2)[$ such that the sequence $(m_n)$ is decreasing and 
satisfies 
\begin{equation}\label{CondmnTLC}
n^{(1 - \alpha d - 2 \beta)/2}\, m_n^{-1}\ =\ o(1)
\end{equation}
for some $\alpha\in](1 - 2(m\eta - 1/m))/d\,,\,1/d[$ and $\beta\in]0\,,\,1/2(d+1)[$.

\noindent
Then, for any $y\in\bkR^d$,
\begin{equation} \label{CLTpchap}
Z_n(y) = \sqrt{n^{1 - \alpha d}}\pab{\pchap(y) - p(y)} \tenddist 
{\mathcal N}\pa{0, \frac{\|K\|_2^2\,p(y)}{1 + \alpha d}} = Z(y)
\end{equation} 
where $\norm{K}_2^2 = \displaystyle\int_{\bkR^d} K^2(t)\,\dt$.
In addition, for $q$ distinct points $y_1, \cdots, y_q$ of $\bkR^d$, we also have 
\begin{equation}\label{multiclt} 
\pab{Z_n(y_1), \cdots, Z_n(y_q)} \tenddist  
\pab{Z(y_1), \cdots, Z(y_q)} 
\end{equation} 
where $Z(y_1), \cdots, Z(y_q)$ are independent. 
\vspace{1ex}

\noindent
Moreover, if [A2bis] holds instead of  [A2], the KDE $\pchap$ satisfies the pointwise and multivariate CLT
if there is a sequence $(v_n)$ of the form $v_n = \eta \log n$\ with $\eta \in]0\,,\,1/2a[$ and $a < m$, thus yielding that 
the condition (\ref{CondmnTLC}) is fulfilled for some $\alpha\in](1 - 2a \eta)/d\,,\,1/d[$.
\end{theorem}

\begin{proof}
The proof is postponed in Appendix C.
\end{proof}

In the following we show that the three families of densities, 
for which we proved the consistency of $\pchap$, also benefit from the CLT.
\vspace{1ex}

\noindent
{\bf Example \ref{Example-dens-dec-pol} (continued)}
We know that $m_n^{-1} = O(n^{\eta\delta})$. 
Hence, condition (\ref{CondmnTLC}) is fulfilled as soon as 
$\alpha d > 1 - 2(\beta - \eta\delta)$ with $\eta < \beta / \delta$.
The two constraints on the bandwidth parameter $\alpha$
are the same if we take $\eta = (\beta + 1/m)/(m+\delta)$ with $\beta\in]\delta/m^2\,,\,1/2(d+1)|$.
Therefore,  the KDE $\pchap$ satisfies the pointwise and multivariate 
CLT for any $\alpha > (1 - 2\beta \tau)/d$ 
with $\tau = (m - \delta / m\beta)/(m + \delta)$.
\vspace{1ex}

\noindent
In this context, the rate of convergence can be improved using the truncated KDE $\pchap^\ast$.
Indeed, from Remark \ref{RemEstTronque}, it is easy to see that  $\pchap^\ast$ satisfies the CLT 
under the conditions (\ref{CondmnTLC})
and $n^{(1 - \alpha d)/2} n^{-(m-1)\eta} = o(1)$.
Thus taking $\eta = \beta / (m+\delta - 1)$, with $\beta < 1/2(d+1)$,
implies that  $\pchap^\ast$ satisfies the CLT
for any $\alpha > (1 - 2 \beta\tau )/d$ with $\tau = (m - 1)/(m+\delta - 1)$.
The pointwise convergence rate of $\pchap^\ast$ is better than that of $\pchap$
since  $(m - 1)/(m+\delta - 1) >\!> (m - \delta/m\beta)/(m+\delta)$.
In addition, constraints on $\delta$ and $\beta$ are really relaxed
since we only have to assume that $\delta > 3$ instead of assuming that $\delta$ is sufficiently large 
to ensure that $\delta < m^2 \beta$.\ $\emptysq$
\vspace{1ex}

\noindent
{\bf Example \ref{Example-dens-dec-exp} (continued)}
The KDE $\pchap$ satisfies the CLT 
for any $\alpha > (1 - 2\beta \tau)/d$ with $\beta\in]0\,,\,1/2(d+1)[$ and
$\tau < 1/2$.\ $\emptysq$

\finl
{\bf Example \ref{Example-dens-gauss} (continued)}
In the Gaussian case, the KDE $\pchap$ satisfies the CLT 
for any $\alpha > (1 - \beta \tau)/d$ with $\tau = (1 - r_f) / (1 - r_f + 1/c)$ 
and $c\in]0,1[$.

\noindent
The best rate of convergence in the CLT is obtained 
when $f$ is bounded. Indeed in this case,  $r_f = 0$ and
$\tau = c / (1 + c)$ with $c\in]0,1[$. 
Therefore the constant $\tau$ involving the condition on the bandwidth parameter 
$\alpha$, can be set as close to $1/2$ as possible. 
This allows the choice of the smaller $\alpha$ value, thus yielding 
to the best convergence rate that can obtained for the KDE.
When $f$ is not bounded and $r_f$ close to 1, the convergence rate obtained in the Gaussian 
case is far from this "best" rate.\ $\emptysq$

\begin{remark}
{\rm In the three previous examples, we have seen that the bandwidth parameter $\alpha$
involving the convergence rate in the CLT, must satisfy
the condition $\alpha > (1 - 2\beta \tau)/d$ with $\tau < 1/2$ and $\beta < 1/2(d+1)$.
Therefore, since $(1 - 2\beta \tau)/d >\!> 1/(d+2)$, 
we clearly have a big loss on the convergence rate in comparison with the sample 
model (ie. $f\equiv0$) or for example with the AR(1) model (ie. $f(x) = \theta x$).
Indeed, in these two cases, we can establish the same CLT for a kernel density estimator of $p$
under the condition $\alpha \in]\frac{1}{d+2}\,,\,\frac{1}{d}[$.  
}\end{remark}

\section{Conclusion}

The strong consistency with rate as well as the pointwise and multivariate 
asymptotic normality have been established for a kernel estimator $\pchap$ of the noise density 
in a functional autoregressive model.
This estimator is based on a predictor sequence of the noise constructed from
a nonparametric estimator $\fchap$ of the autoregression function. The properties of the 
KDE depend on the convergence  of the average prediction error of $\fchap$, which 
hampers the study of $\pchap$ and the clarifying of its convergence rate. Indeed, 
it requires the knowledge of the way the density decreases at infinity as well as
the choice of a well-suited sequence which controls the convergence of the prediction errors of
$\fchap$. We have at least exhibited three families of densities which are consistently estimated 
with $\pchap$ and benefit from the central limit theorem.
Nevertheless from a practical point of view, the convergence results ensure
that the KDE $\pchap$ may behave pretty well in many situations.

\section*{Appendix A}
\renewcommand{\thesection}{\arabic{section}} 
\renewcommand{\thesection}{\Alph{section}} 
\renewcommand{\theequation}{\thesection.\arabic{equation}} 
\newtheorem{lemapp}{Lemma}[section] 
\newtheorem{theoapp}[lemapp]{Theorem} 
\newtheorem{corolapp}[lemapp]{Corollary} 
\setcounter{section}{1} 
\setcounter{equation}{0} 
\setcounter{lemapp}{0} 
\def\calF{{\cal F}}

This appendix is devoted to two technical lemmas useful 
in the different proofs of the paper.

\begin{lemapp}\label{LemmaMajSomInd}
Let $(Z_n)$ be a sequence of positive real random variables and
let $(v_n)$ be a sequence of positive real numbers increasing to infinity.
Assume that $\sum_{i=1}^{n} g(Z_i) = O(n)$ a.s. for some 
increasing positive function $g:\bkR \rightarrow \bkR$.
Let $b \geq 0$. 
Then, if the sequence $(v_n^b / g(v_n))$ is decreasing,
\begin{equation*}
\sum_{i=1}^{n} Z_i^b\,\ind_{\acc{Z_i > v_i}} = O\Bigl(\sum_{i=1}^{n} \frac{v_i^b}{g(v_i)}\Bigr)\quad\mbox{a.s.}
\end{equation*}
Moreover, if $(a_n)$ is a sequence of positive real numbers increasing to infinity
such that the sequence $(a_n / g(v_n))$ is decreasing, then
\begin{equation*}
\sum_{i=1}^{n} a_i\,\ind_{\acc{Z_i > v_i}} = O\Bigl(\sum_{i=1}^{n} \frac{a_i}{g(v_i)}\Bigr)\quad\mbox{a.s.}
\end{equation*}
\end{lemapp}

\begin{proof}
The proof is very standard and is therefore omitted here.
The proof of the first part can be found in \cite{Portier-Oulidi} for example.
The proof of the second part is a direct adaptation of the first one.
\end{proof}

For all the sequel, let us denote by $\calF_n$ the $\sigma$-algebra of the events 
occurring up to time $n$, ie. $\calF_n = \sigma(X_0, \veps_1, \cdots,\veps_n)$.
\vspace{1ex} 
 
\begin{lemapp}\label{LemmaG}  
In the context of model (\ref{Model}), assume that [A1]--[A4] hold true.
Let $(U_n)_{n\geq 0}$ be a sequence of random vectors adapted to the filtration $(\calF_n)_{n\geq 0}$.
For any $x\in\bkR^d$ and $n\geq 1$, let us define
$$G_n(x) = \sum_{i=1}^{n} i^b \Big(K\pa{i^a (X_i - U_{i-1} - x)} - i^{- ad}p(x + U_{i-1} - f(X_{i-1})) \Big)$$
where $K$ is a kernel satisfying [A5], $b \in ]0,1[$ and $a \in ]0,1/d[$.
Then, for any constants $A > 0$ and $\nu > 0$, 
\begin{equation}\label{ResGn} 
\sup_{\norm{\,x\,}\leq A n^\nu} \abs{G_n(x)}= o\pa{n^s} + O\pa{n^{1 + b - ad - a}}\hspace{1cm}\mbox{a.s.} 
\end{equation} 
where $s\in](1+ 2b - ad)/2,1[$. 

\noindent
In addition, if $\sum_{i=1}^{n} \|U_{i-1} - f(X_{i-1})\| = o(n)$ a.s. and if $a > 1/(d+2)$, then for any $x\in\bkR^d$,
\begin{equation}\label{ResGnCLT}
n^{-(1 + 2b - ad)/2}\ G_n(x)\  \tenddist\ 
{\mathcal N}\pab{0, (1 + 2b - ad)^{-1}\|K\|_2^2\,p(x)} = G(x)
\end{equation} 
and for two distinct points $x,y$ of $\bkR^d$,
\begin{equation}\label{ResGnMultiCLT}
n^{-(1 + 2b - ad)/2}
\pa{ \begin{array}{l} G_n(x) \\ G_n(y)\\  \end{array}} \ \liml\ 
   \pa{ \begin{array}{l} G(x) \\ G(y)\\ \end{array}} 
\end{equation} 
where $G(x)$ and $G(y)$ are independent.
\end{lemapp} 
 
\begin{proof}
For any $y\in\bkR^d$, let us denote $K_i(y) = K(i^a y)$ and let us decompose 
$G_n(x)$ under the form $G_n(x) = M_n(x) + R_n(x)$
where
\begin{eqnarray*} 
M_n(x) &=& \sum_{i=1}^{n} i^b \Bigl(K_i\pa{X_i - U_{i-1} - x}
- \bkE\cro{K_i\pa{X_i - U_{i-1} - x}  | \calF_{i-1}}\Bigr)\\
R_n(x) &=& \sum_{i=1}^{n} i^b \Bigl(\bkE\cro{K_i\pa{X_i - U_{i-1} - x}  | \calF_{i-1}}
- i^{- ad}p(x + U_{i-1} - f(X_{i-1}))  \Bigr)
\end{eqnarray*} 
For any $x\in\bkRd$, $(M_n(x))$ is a square integrable real martingale.
Its increasing process $(\CMx_n)$ is defined by 
\begin{eqnarray}\label{MnCrochet} 
\CMx_n & = & \sum_{i=1}^{n} i^{2b}\bkE\crob{K_i^2\pa{X_i - U_{i-1} - x} | \calF_{i-1}}\nonumber\\
& & \  - \sum_{i=1}^{n} \pab{i^{b} \bkE\cro{K_i\pa{X_i - U_{i-1}  - x} | \calF_{i-1}}}^2 
\end{eqnarray} 
First of all, for all $x\in\bkRd$, set $\Delta M_{n}(x)=M_{n}(x)-M_{n-1}(x)$. 
Since the kernel $K$ is bounded and Lipschitz,  
for all $\delta\in]0\,,\,1[$, one can find some 
positive constant $C_\delta$ such that, for any $x, y \in\bkR^d$ 
\begin{equation}\label{MajK} 
\abs{K(x) - K(y)}  \leq   C_\delta \norm{x - y}^{\delta}. 
\end{equation} 
Now, following  the same development and using the same argues as for the proof 
of Lemma\,A.2 in  Bercu and Portier \cite{Bercu2008},
there are positive constants $c_1,c_2,c_3$ and $c_4$ such that for all $n \geq 1$,
$\abs{\Delta M_n(0)} \leq  c_1 n^b$ and
$<\!M(0)\!>_n \leq c_2 n^{1 +  2b -ad}$, and for any $x, y \in\bkR^d$ and any $\delta\in]0\,,\,1[$,
\begin{eqnarray*} 
\abs{\Delta M_n(x) - \Delta M_n(y)} & \leq & c_3 \norm{x - y}^{\delta}\,n^{b + a\delta}\\
<\! M(x)\!-\!M(y)\!>_n & \leq & c_4  \norm{x - y}^{2\delta} n^{1+2b - ad +2a\delta}. 
\end{eqnarray*} 
Finally, as the power $\delta$ can be chosen as small as one wishes, all the four conditions of 
Theorem\,6.4.34 of \cite{Duflo} are fulfilled.
Consequently, for any constants $A > 0$ and $\nu > 0$, 
\begin{equation} 
\label{ResGn_Mn} 
\sup_{\norm{\,x\,}\leq A n^\nu} \abs{M_n(x)}= o\pa{n^s} \hspace{1cm}\mbox{a.s.} 
\end{equation} 
for any $s\in ](1+ 2b - ad)/2,1[$.

\noindent
Let us now study the remainder $R_n(x)$.
For any $i\geq 1$, we have
$$\bkE\cro{K_i\pa{X_i - U_{i-1} - x}  | \calF_{i-1}}=
\int_{\bkR^d} K\pab{i^a (v + f(X_{i-1}) - U_{i-1} - x)} p(v) \dv,$$
and via the change of variables $t= i^a(v + f(X_{i-1})  - U_{i-1} - x)$, we deduce 
that
\begin{equation*} 
R_n(x) = \sum_{i=1}^{n} i^{b- a d}\int_{\bkR^d} K(t)
\left(p(i^{-a}t + x + U_{i-1} - f(X_{i-1})) -  p(x + U_{i-1} - f(X_{i-1}))\right)\dt.
\end{equation*} 
Then, since by [A4] the gradient of $p$ is bounded, we deduce
by a Taylor expansion that 
\begin{equation}
\label{ResGn_Rn} 
\sup_{x\in\mathbb{R}^d}\abs{R_n(x)} =O\pa{n^{1 + b - ad - a}} \hspace{1cm}\mbox{a.s.} 
\end{equation} 
which, combined with (\ref{ResGn_Mn}) ends the proof of (\ref{ResGn}).

\noindent
Let us now establish (\ref{ResGnCLT}). 
From (\ref{ResGn_Rn}), we immediately deduce that
\begin{equation}
\label{ResGn_RnCLT} 
n^{-(1 + 2b - ad)/2} \sup_{x\in\mathbb{R}^d}\abs{R_n(x)} = o(1) \hspace{1cm}\mbox{a.s.} 
\end{equation} 
as soon as $a > 1/(d+2)$.
Therefore, it remains to establish that
\begin{equation}\label{ResGn_MnCLT}
n^{-(1 + 2b - ad)/2}\ M_n(x)\  \tenddist\ 
{\mathcal N}\pab{0, (1 + 2b - ad)^{-1}\|K\|_2^2 p(x)}
\end{equation} 
In order to make use of the CLT for martingales 
(see e.g. \cite{Duflo}, Corollary 2.1.10, p.46), we have to study
the asymptotic behaviour of $<\!M(x)\!>_n$ and to check 
that Lindeberg's condition is satisfied.
\vspace{1ex}

Starting from (\ref{MnCrochet}), we can rewrite $\CMx_n$ under the form 
\begin{equation}\label{DevMnCrochet} 
\CMx_n = A_n(x) + p(x) \norm{K}_2^2 \sum_{i=1}^{n} i^{2b - ad} + o\pa{n^{1 + 2b - ad}}
\end{equation} 
where
\begin{equation*}
A_n(x) = \sum_{i=1}^{n} i^{2b- a d}\int_{\bkR^d} K^2(t)
\left(p(i^{-a}t + x + U_{i-1} - f(X_{i-1})) -  p(x))\right)\dt
\end{equation*}
Using one more time the fact that the gradient of $p$ is bounded, we deduce that
\begin{equation*}
\abs{A_n(x)} = O(n^{1+ 2b- a d - a}) + n^{2b- a d}\,O\pa{\sum_{i=1}^{n} \norm{U_{i-1} - f(X_{i-1})} }
\end{equation*}
Finally, since $n^{-(1 + 2b - ad)}(1 + 2b - ad)\sum_{i=1}^{n} i^{2b - ad} \tendvers 1$
and we have assumed that  $\sum_{i=1}^{n} \norm{U_{i-1} - f(X_{i-1})} = o(n)$, we obtain that
\begin{equation} 
n^{-(1 + 2b - ad)} \CMx_n\ \tendas\ (1 + 2b - ad)^{-1} \norm{K}_2^2\,p(x)
\end{equation} 
which defines the variance in the CLT.
The Lindeberg's condition is easily obtained following the proof
of Lemma B.1 in \cite{Bercu2008}.

\noindent
To establish (\ref{ResGnMultiCLT}), it is enough to prove
\begin{equation}
n^{-(1 + 2b - ad)/2}
\pa{ \begin{array}{l} M_n(x) \\ M_n(y)\\  \end{array}} \ \liml\ 
   \pa{ \begin{array}{l} G(x) \\ G(y)\\ \end{array}} 
\end{equation} 
which, following the same lines as in \cite{Bercu2008}, 
consists in showing that 
\begin{equation} 
\label{cvgnullMxy} 
\lim_{\ntinf} n^{-(1 + 2b - ad)}\,\sum_{i=1}^{n}  
\mathbb{E}\cro{\Delta M_i(x) \Delta M_i(y)|\calF_{i-1}} = 0 \hspace{1cm}\mbox{a.s.} 
\end{equation} 
For all $i \geq 1$, we have 
\begin{eqnarray*} 
\mathbb{E}\cro{\Delta M_i(x) \Delta M_i(y)|\calF_{i-1}}
& \leq & i^{2b} \bkE\cro{K_i \pa{X_i - U_{i-1} -  x} K_i\pa{X_i - U_{i-1} -  y} | \calF_{i-1}}\\ 
& \leq & 
i^{2b - ad}\int_\bkRd K(t)\, K\pa{t+i^a(x - y)} \times \\
& & \hspace{1.4cm} p\pa{i^{-a} t + x + U_{i-1} - f(X_{i-1})}\,\dt. 
\end{eqnarray*}
Therefore, as the gradient of $p$ is bounded and
$\sum_{i=1}^{n} \norm{U_{i-1} - f(X_{i-1})} = o(n)$, we derive that 
\begin{equation*} 
\sum_{i=1}^{n} \mathbb{E}\cro{\Delta M_i(x) \Delta M_i(y)|\calF_{i-1}} 
\leq 
Q_n(x,y)+O\pab{n^{1 +2b - ad - a}} + o\pab{n^{1+2b-ad}}\hspace{0.5cm}\mbox{a.s.} 
\end{equation*} 
where\ $\displaystyle Q_n(x,y) =\sum_{i=1}^{n} i^{2b-ad} p(x)\int_\bkRd\! K(t)K(t + i^a (x-y))\dt$.  

\noindent
However, using the fact that $K$ is compactly supported, we deduce that for $i$ 
large enough, the integral term in  $Q_n(x,y)$ is zero, which yields  (\ref{cvgnullMxy}). 
\end{proof}

\section*{Appendix B}
\setcounter{section}{2} 
\setcounter{equation}{0} 
\def\Mtild{\widetilde M}
\def\Rtild{\widetilde R}
\def\Gdeuxn{G_{2,n}}
\def\Gtroisn{G_{3,n}}
\def\Gunn{G_{1,n}}

This appendix is concerned with the proof of Theorem\,\ref{cvud.fchap}.
As already mentioned in paragraph 2.3, the uniform convergence over dilated sets of $\fchap$
has already been studied by Portier and Oulidi \cite{Portier-Oulidi} in a controlled framework, ie. for
model of the form $X_{n+1} = f(X_n) + U_n + \varepsilon_{n+1}$. 
We follow the same steps of the proof, adapt and a little bit improve 
the results in the context of model (\ref{Model}).
Starting from the definition of $\fchap (x)$, we can write
\begin{eqnarray}
\label{decomp-fn-f}
\fchap(x) - f(x)\ =\ \frac{\Mtild_n(x) + \Rtild_{n-1}(x)}{H_{n-1}(x)}
\ind_{\acc{H_{n-1}(x)\not = 0}}\ -\ 
f(x)\ind_{\acc{H_{n-1}(x) = 0}}
\end{eqnarray}
\begin{tabular}{ll}
with & $\displaystyle
\Mtild_n(x)~=~\sum_{i=1}^{n-1}\,i^{\beta d}\,K\pab{i^\beta(X_i - x)}\varepsilon_{i+1}$\\
& $\displaystyle
\Rtild_{n-1}(x)~=~\sum_{i=1}^{n-1}\,i^{\beta d}\,K\pab{i^\beta(X_i - x)}
\pab{f(X_i) - f(x)}$\\
& $\displaystyle
H_{n-1}(x)~=~\sum_{i=1}^{n-1}\,i^{\beta d}\,K\pab{i^\beta(X_i - x)}$
\end{tabular}

\noindent
$\bullet$ Although the process $(X_n)$ is not ruled by the same equation, $\Mtild_n(x)$ 
is treated as in \cite{Portier-Oulidi}.
So, we only give the result, that is, for any\ $A < \infty$ and $\nu > 0$,
and for all  $s > 1/2\,+\, \beta d$ with $\beta < 1/2d$,
\begin{eqnarray}\label{resfchap1}
\sup_{\|x\| \leq\,A\,n^\nu} \norm{\Mtild_n(x)}\ \egalas\ o(n^{s}) .
\end{eqnarray}
$\bullet$ Let us now study $\Rtild_n$.
Since by [A3] the gradient of  $f$ is bounded and kernel $K$ has compact support,
there is a positive constant $c$ such that $K(i^\beta (X_i - x)) \| f(X_i) - f(x) \| 
\leq c \, i^{-\beta} K(i^\beta (X_i - x))$. Thus, it follows that
\begin{eqnarray}
\norm{\Rtild_n(x)} & = & O\pa{\sum_{i=1}^{n}\,i^{\beta d - \beta}\,K\pab{i^\beta(X_i - x)}} \label{major-Rntild}\\
& = & O\pa{\Gunn(x) + \sum_{i=1}^{n} \,i^{- \beta} p\pab{x - f(X_{i-1})}} \nonumber
\end{eqnarray}
where $\Gunn(x)$ stands for $G_n(x)$ in Lemma\,\ref{LemmaG}
with $b = \beta d - \beta$, $a = \beta$ and $U_i = 0$.
Now, since $p$ is bounded and using Lemma\,\ref{LemmaG},
we obtain that, for any\ $A < \infty$ and $\nu > 0$, and for all $s^{\prime} > (1 + \beta d - 2\beta)/2$, 
\begin{eqnarray}\label{resfchap2}
\sup_{\|x\| \leq\,A\,n^\nu} \norm{\Rtild_n(x)}\ \egalas\ o(n^{s^{\prime}}) + O(n^{1 - \beta})\ =\ O(n^{1 - \beta}).
\end{eqnarray}
%

\noindent
$\bullet$ The term $H_{n-1}(x)$ remains to be studied.
To this aim, let us write $H_n(x)$ in the form $\Gdeuxn(x) + \sum_{i=1}^{n} p\pab{x - f(X_{i-1})}$
where $\Gdeuxn(x)$ stands for $G_n(x)$ in Lemma\,\ref{LemmaG}
with $b = \beta d$, $a = \beta$ and $U_i = 0$.

Let $R$ be such that $\int\norm{f(t)}h(t)\dt < R < \infty$ and let $(v_n)$ be a sequence of positive real
numbers increasing to infinity such that $v_n = O(n^\nu)$ for any $\nu > 0$. 
For $x\in\bkR^d$ such that $\norm{x} \leq v_n$, we have
\begin{eqnarray}\label{ResHn3}
\norm{p}_\infty\ \geq\ \frac{1}{n}\,\sum_{i=1}^{n} p\pab{x - f(X_{i-1})}\ \geq\ \frac{m_n}{n} \sum_{i=0}^{n-1}\ind_{\acc{\norm{f(X_i)} \leq R}}
\end{eqnarray}
which, combined with (\ref{liminf-ball}), yields that 
\begin{eqnarray}\label{Resinfp}
\liminf_{\ntinf} \inf_{\norm{x} \leq v_n} \frac{1}{n m_n}\sum_{i=1}^{n} p\pab{x - f(X_{i-1})} \ > \ 0  \hspace{1cm}\mbox{a.s.}
\end{eqnarray}
In addition, using Lemma\,\ref{LemmaG} applied to $\Gdeuxn(x)$
and the fact that by [A6], $m_n^{-1} = \inf\pa{O(n^{1-\lambda}),o(n^\beta)}$,
we derive that for any positive constants $A$ and $\nu$,
\begin{equation}\label{ResGdeuxn}
\sup_{\norm{\,x\,}\leq A n^\nu} \frac{\abs{\Gdeuxn(x)}}{n m_n} = o(1)\quad\mbox{a.s.}
\end{equation}
Now, since
\begin{equation*}
\inf_{\norm{x} \leq v_n} \frac{H_{n-1}(x)}{n m_n} \  >\
\inf_{\norm{x} \leq v_n} \frac{1}{n m_n}\sum_{i=1}^{n} p\pab{x - f(X_{i-1})}
- \sup_{\norm{\,x\,}\leq A n^\nu} \frac{\abs{\Gdeuxn(x)}}{n m_n}
\end{equation*}
we infer using (\ref{Resinfp}) together with (\ref{ResGdeuxn}), that 
\begin{eqnarray}\label{ResHn}
\liminf_{\ntinf} \frac{1}{n m_n} \inf_{\norm{x} \leq v_n} H_{n-1}(x)\ > \ 0  \hspace{1cm}\mbox{a.s.}
\end{eqnarray}
Finally, from (\ref{decomp-fn-f}), unifying the results obtained for $\Mtild_n(x)$, $\Rtild_{n}(x)$ and $H_n(x)$ 
completes the proof of the first part of Theorem\,\ref{cvud.fchap}.

\noindent
The second part immediately follows by noting that
when $\beta < 1 / 2(d+1)$, we can choose $\lambda\in]1/2 + \beta d\,,\,1[$ 
such that $\lambda = 1 - \beta$.

\begin{remark}{\rm\ From the proof of Theorem\,\ref{cvud.fchap}, 
we deduce an interesting upper bound of the estimation error: for all $x$ in $\bkR^d$ and $n \ge 1$,
\begin{equation}
\label{majfchap}
\| \fchap (x) - f(x) \| \le C_{f,K} + \norm{x} + \varepsilon^{\sharp}_n \quad \mbox{a.s.}
\end{equation}
where $\displaystyle\varepsilon^{\sharp}_n:=\sup_{i \le n} \| \varepsilon_i \|$ and $C_{f,K}$ is a constant
depending on $K$ and $f$. This upper bound will be useful for the study of the average prediction error of $f$ in the proof of 
Corollary \ref{CorErrPred} for example. 
It directly follows from the decomposition (\ref{decomp-fn-f}); 
the constant $c_f$ comes from the study of $\Rtild_{n-1}(x)/H_{n-1}(x)$ and (\ref{major-Rntild}),
the term $\varepsilon^{\sharp}_n$ from the study of $\Mtild_n(x)/H_{n-1}(x)$.
}\end{remark}

\section*{Appendix C}
\setcounter{section}{3} 
\setcounter{equation}{0} 

This appendix is devoted to proving Theorem\,\ref{th.sur.p} and 
Theorem \ref{TheoCLTpchap}.
Starting from (\ref{defpchap}) we infer that 
for all $y\in\bkRd$ and $n\geq 1$, 
\begin{equation}\label{Decomppchap}
\pchap(y) - p(y)\ =\ \frac{1}{n}\pab{\Gtroisn(y) + B_n(y)}
\end{equation}
where\ $\displaystyle B_n(y)\,=\,\sum_{i=1}^{n} \pab{p(y + \widehat f_{i-1}(X_{i-1}) - f(X_{i-1}))-p(y)}$\ 
and $\Gtroisn$ stands for $G_n$ in Lemma\,\ref{LemmaG} 
with $b = \alpha d$, $a=\alpha$ and $U_i = \widehat f_i(X_i)$.
%
\paragraph{Proof of Theorem \ref{th.sur.p}}

\noindent
$\bullet$ Since the gradient of $p$ is bounded, we deduce that
\begin{eqnarray}\label{MajBn} 
\sup_{y\in\mathbb{R}^d}\abs{B_n(y)} & = & O\left(\sum_{i=0}^{n-1}\norm{\widehat f_i(X_i) - f(X_i)} \right) 
\hspace{1cm}\mbox{a.s.}
\end{eqnarray} 
which, combined with Corollary\,\ref{CorErrPred}, leads to
\begin{equation}\label{MajBnbis}
\frac{1}{n} \sup_{y\in\mathbb{R}^d}\abs{B_n(y)} = O\pa{\frac{n^{-\beta}}{m_n}} + O\pa{w_n}\quad\mbox{a.s.}
\end{equation}
where, in the context of Assumption [A2] with a sequence $(v_n)$ of the form $A n^\eta$ with $A > 0$ and $\eta >0$, 
$w_n = n^{- m\eta + 1 /m}$  if $\eta < 1/m$ and 
$w_n = n^{-\eta(m-1)}$ otherwise. 
As we assumed that $m> 2$, $\eta > 1/m^2$ and $m_n^{-1} = o(n^\beta)$,
we clearly obtain that $\sup_{y\in\mathbb{R}^d}\abs{B_n(y)} = o(n) $ a.s.

In the framework of [A2bis], Corollary\,\ref{CorErrPred} used with $v_n = \eta \log n$ where $\eta >0$,
leads to (\ref{MajBnbis}) with $w_n = (\log n) n^{-a \eta}$ for $a < m$, and thus also implies that
 $\sup_{y\in\mathbb{R}^d}\abs{B_n(y)} = o(n) $ a.s.

\vspace{1ex}

\noindent
$\bullet$ 
Applying Lemma\,\ref{LemmaG} on  $\Gtroisn$, we obtain that for any $A > 0$ and $\nu > 0$,
\begin{equation}\label{Maj-Gtroisn} 
\frac{1}{n}\sup_{\norm{\,y\,}\leq A n^\nu} \abs{\Gtroisn(y)}
= o\pa{n^{\gamma - 1}} + O\pa{n^{- \alpha}} 
\hspace{1cm}\mbox{a.s.} 
\end{equation} 
where $\gamma\in](1+\alpha d)/2,1[$. 
Then, combining (\ref{MajBnbis}) together with (\ref{Maj-Gtroisn}) applied with $A=2$ and $\eta=1/2$, we obtain that almost surely
\begin{equation}\label{Maj-pchap-p} 
\sup_{\norm{\,y\,}\leq 2 \sqrt{n}} \abs{\pchap(y) - p(y)}
= o\pa{n^{\gamma - 1}} + O\pa{n^{- \alpha}} + O\pa{\frac{n^{-\beta}}{m_n}} + O\pa{w_n}
\end{equation} 
To close the proof of Theorem\,\ref{th.sur.p}, let us show that 
\begin{equation} 
\label{Respchap2} 
\sup_{\norm{y} > 2\sqrt{n}} \abs{\pchap(y) - p(y)}
= O\pa{\frac{1}{n}} \hspace{1cm}\mbox{a.s.} 
\end{equation} 
Under [A2] (and so [A2bis]), we know by (\ref{MajSumetSupNormX}) that almost surely 
$\sup_{i\leq n}\norm{X_i} = O\pa{\veps_n^\#} = o\pa{n^{1/m}}$ 
with $m > 2$.
Moreover, thanks to the bound (\ref{majfchap}) on $\fchap(x) - f(x)$, we derive that 
$\sup_{i\leq n}\norm{\widehat f_i(X_i)} = O\pa{1+\veps_n^\# + \sup_{i\leq n}\norm{X_i}}$ a.s.
Thus, it follows that
\begin{equation*} 
\sup_{i\leq n}\norm{X_i - \widehat f_{i-1}(X_{i-1})} = o(n^{1/m}) = o\pa{\sqrt{n}}\hspace{1cm}\mbox{a.s.} 
\end{equation*}
Hence, for $n$ large enough, $\norm{X_i - \widehat f_{i-1}(X_{i-1})} < \sqrt{n}$ a.s.  for any $i \leq n$,  
which assures that, for $y$ such that $\norm{y} > 2\sqrt{n}$,  
$\norm{X_i - \widehat f_{i-1}(X_{i-1}) - y} > \sqrt{n}$ a.s.      
Therefore, since $K$ is compactly supported, it clearly leads to 
\begin{equation} 
\sup_{\|y\| > 2\sqrt{n}}\abs{\sum_{i=1}^{n} i^{\alpha d}
K\pa{i^\alpha(X_i - \widehat f_{i-1}(X_{i-1}) - y)}} = O\pa{1}\quad\mbox{a.s.}  
\end{equation} 
and
\begin{equation} 
\label{ResConvUnif3} 
\sup_{\normp{y} > 2\sqrt{n}}\abs{\pchap(y)} = O\pa{\frac{1}{n}}\quad\mbox{a.s.}  
\end{equation} 
In addition, since $(\veps_n)$ has a finite moment of order $m>2$ and $p$ is positive, 
it follows that $p(y) = O(\norm{y}^{-3})$ for large values of $y$, leading to  
\begin{equation}\label{ResConvUnif4} 
\sup_{\normp{y} > 2\sqrt{n}} p(y) = O\pab{\frac{1}{n}}. 
\end{equation} 
Consequently, (\ref{Respchap2}) is deduced 
from (\ref{ResConvUnif3}) and (\ref{ResConvUnif4}),
which achieves the proof of Theorem \ref{th.sur.p}.

\paragraph{Proof of Theorem \ref{TheoCLTpchap}}
From the decomposition (\ref{Decomppchap}), we deduce that for any $y\in\bkR^d$,
\begin{equation}
\sqrt{n^{1-\alpha d}}\,\pa{\pchap(y) - p(y)}\ =\ n^{-(1 + \alpha d)/2}\, \Gtroisn(y) + n^{-(1 + \alpha d)/2} B_n(y)
\end{equation}
Thanks to the second part of Lemma\,\ref{LemmaG}, we derive that
for any $\alpha \in ]\frac{1}{d+2}\,,\,\frac{1}{d}[$, 
\begin{equation} 
n^{-(1 + \alpha d)/2}\, \Gtroisn(y)  \tenddist 
{\mathcal N}\pab{0, (1 + \alpha d)^{-1}\|K\|_2^2\, p(y)}
\end{equation} 
Therefore, to establish the pointwise CLT for $\pchap$,
we only have to prove that 
\begin{equation}\label{CondBnpourCLT} 
n^{-(1 + \alpha d)/2}\, \abs{B_n(y)} = o(1) \quad\quad\mbox{a.s.}
\end{equation} 
Using (\ref{MajBnbis}) and the fact that $\eta < 1/m$, condition (\ref{CondBnpourCLT}) reduces to (\ref{CondmnTLC})
%
and, in the context of [A2], 
\begin{equation}\label{CondsurBn2} 
n^{(1 - \alpha d)/2}\, n^{- m\eta + 1 /m} = o(1)
\end{equation} 
With a value of $\eta$ chosen in $]1/m^2\,,\,1/m^2 + 1/m(d+2)[$,
condition (\ref{CondsurBn2}) is fulfilled as soon as the bandwidth parameter $\alpha$ 
satisfies $\alpha d > 1 - 2(m\eta - 1/m)$.

\noindent
In the framework of [A2bis], condition (\ref{CondsurBn2}) is replaced by
$n^{(1 - \alpha d)/2}\, n^{- a\eta} \log n= o(1)$ with $a < m$.
Therefore, with a value of $\eta$ chosen in $]0\,,\, 1/2a[$, 
it is fulfilled as soon as $\alpha d > 1 - 2 a\eta$,
which ends the proof of the first part of Theorem \ref{TheoCLTpchap}. 

\noindent
The multivariate CLT remains to be proven. 
Taking the previous results into account, it is enough to establish that for  
two distinct points $x,y \in \bkR^d$,
\begin{eqnarray*} 
   n^{-(1+\alpha d)/2}\left( \begin{array}{ll} 
    \Gtroisn(x) \\ 
    \Gtroisn(y) 
   \end{array} \nonumber \right)\ \liml\ 
{\cal N}\pa{0\,,\,\frac{\|K\|_2^2}{1 + \alpha d}
\pa{\begin{array}{cc} 
    p(x) & 0 \\ 
    0 & p(y) \\
\end{array}}  }
\end{eqnarray*} 
This result is straightforwardly given by Lemma \ref{LemmaG},
which closes the proof of Theorem \ref{TheoCLTpchap}.\ $\emptysq$


\end{document}